\documentclass[envcountsame]{llncs}
\usepackage{mathscinet}
\usepackage{comment}
\usepackage{graphicx}
\usepackage{booktabs}
\usepackage[hidelinks]{hyperref}
\usepackage{enumerate}
\usepackage[ddmmyyyy]{datetime}
\usepackage{subcaption}
\usepackage{textcomp}		
\usepackage{paralist}

\usepackage[dvipsnames]{xcolor}
\usepackage{tikz}
\usetikzlibrary{fit}
\usetikzlibrary{intersections}
\usetikzlibrary{decorations.pathreplacing}
\usetikzlibrary{calc}

\newcommand{\link}[4]{
\pgfmathsetmacro\middle{#1 / 2 + #3 / 2}
\draw [very thick] plot [smooth, tension=1.8] coordinates {(#1,#2) (\middle,#4) (#3,#2)} ;
}

\newcommand{\altlink}[4]{
\pgfmathsetmacro\middle{#1 / 2 + #3 / 2}
\draw [very thick, color=red] plot [smooth, tension=1.8] coordinates {(#1,#2-0.5) (\middle,-#4-0.5) (#3,#2-0.5)} ;
}

\newcommand{\linkb}[4]{
\pgfmathsetmacro\middle{#1 / 2 + #3 / 2}
\draw [very thick] plot [smooth, tension=1.8] coordinates {(#1,0) (#2,#4) (#3,0)} ;
}

\usepackage[vlined, ruled, linesnumbered]{algorithm2e}
\usepackage{algorithmicx}

\usepackage{amsmath}
\usepackage{amssymb}

\usepackage{xspace}
\usepackage{authblk}
\usepackage[color=green!20]{todonotes}

\newtheorem{observation}[theorem]{Observation}

\newcommand{\designext}{\textsc{RNA Design Extension}\xspace}

\def\3sat{\textsc{3-Sat}\xspace}

\newcommand{\subtree}{subtree\xspace}

\renewcommand\leq\leqslant
\renewcommand\geq\geqslant

\newcommand{\lit}[1]{L\langle #1 \rangle}

\newcommand{\var}[1]{V\langle #1 \rangle}
\newcommand{\cla}[1]{S\langle #1 \rangle}

\newcommand{\litalt}[2]{L_{-#1y}\langle #2 \rangle}

\usepackage[utf8]{inputenc}
\usepackage[T1]{fontenc}
\usepackage{a4wide}

\excludecomment{full-version}
\includecomment{conference-version}

\newif\ifconf

\newcommand{\appendixproof}[2]{#2}

\title{Designing RNA Secondary Structures is Hard}	

\author{\'Edouard Bonnet\inst1 \and Pawe\l{} Rz\k{a}\.zewski \inst2 \and Florian Sikora \inst{3}}

\institute{Middlesex University, Department of Computer Science, London, UK \email{edouard.bonnet@dauphine.fr} \and
Faculty of Mathematics and Information Science,		Warsaw University of Technology, Warsaw, Poland \email{p.rzazewski@mini.pw.edu.pl} \and
Universit\'{e} Paris-Dauphine, PSL Research University, CNRS, LAMSADE, Paris, France \email{florian.sikora@dauphine.fr}
}

\date{}

\begin{document}

\conffalse

\maketitle

\begin{abstract}
  An RNA sequence is a word over an alphabet on four elements $\{A,C,G,U\}$ called bases.
  RNA sequences fold into secondary structures where some bases pair with one another while others remain unpaired.
  Pseudoknot-free secondary structures can be represented as well-parenthesized expressions with additional dots, where pairs of matching parentheses symbolize paired bases and dots, unpaired bases.
  The two fundamental problems in RNA algorithmic are to \emph{predict} how sequences fold within some model of energy and to \emph{design} sequences of bases which will fold into targeted secondary structures.
  Predicting how a given RNA sequence folds into a pseudoknot-free secondary structure is known to be solvable in cubic time since the eighties and in truly subcubic time by a recent result of Bringmann et al. (FOCS 2016), whereas Lyngs\o{} has shown it is NP-complete if pseudoknots are allowed (ICALP 2004).
As a stark contrast, it is unknown whether or not designing a given RNA secondary structure is a tractable task; this has been raised as a challenging open question by Anne Condon (ICALP 2003).
Because of its crucial importance in a number of fields such as pharmaceutical research and biochemistry, there are dozens of heuristics and software libraries dedicated to RNA secondary structure design.
It is therefore rather surprising that the computational complexity of this central problem in bioinformatics has been unsettled for decades. 

In this paper we show that, in the simplest model of energy which is the Watson-Crick model
the design of secondary structures is NP-complete if one adds natural constraints of the form: \emph{index $i$ of the sequence has to be labeled by base $b$}.
This negative result suggests that the same lower bound holds for more realistic models of energy. 
It is noteworthy that the additional constraints are by no means artificial: they are provided by all the RNA design pieces of software and they do correspond to the actual practice (see for example the instances of the EteRNA project).
Our reduction from a variant of \textsc{3-Sat} has as main ingredients: arches of parentheses of different widths, a linear order interleaving variables and clauses, and an intended \emph{rematching strategy} which increases the number of pairs iff the three literals of a same clause are false.
The correctness of the construction is also quite intricate; it relies on the polynomial algorithm for the design of saturated structures -- secondary structures without dots -- by Hale{\v s} et al. (Algorithmica 2016), counting arguments, and a concise case analysis.
\end{abstract}

\section{Introduction}\label{sec:intro}
Ribonucleic acid (RNA) is a molecule playing an important role besides deoxyribonucleic acid (DNA) and proteins.
RNA is a chain of nucleotides (or bases) and can be represented as a sequence on a 4-letter alphabet: $A, U, C, G$; denoting the first letter of the corresponding base.
Unlike DNA, RNA is single stranded, and \emph{folds} into itself: some of the bases are linked to each other (they are \emph{paired} or \emph{matched}) to form a stable and compact structure.
This pairing forms the \emph{secondary structure} of the RNA molecule; the primary structure is the sequence of nucleotides and the tertiary structure is the 3D shape.
Predicting how an RNA molecule folds is vital to understand its biological function.

Experiments reveal that the secondary structure of an RNA strand tends to follow the laws of thermodynamics.
Given a model associating a free-energy value to secondary structures, it is widely accepted, since the pioneer work of the chemistry Nobel laureate Christian B. Anfinsen~\cite{Anfinsen1973}, that the secondary structure of a sequence can be predicted as the one with the minimum free-energy (MFE), i.e., the one ensuring the greatest stability. 
The most simple energy model is the Watson-Crick model, allowing $A$ to pair with $U$ and $C$ to pair with $G$ (it also can be seen as the Nussinov-Jacobsen model using only AU and GC base pairs).
In this model, the MFE is simply realized by a structure with the greatest number of pairs. 

\paragraph*{RNA folding.}
A stem-loop or hairpin loop is a building block of RNA secondary structures.
It consists of a series of consecutive base pairs (called double-helix or stackings) ending in a loop of unpaired nucleotides.
A pseudoknot occurs when some nucleotides of this loop pair somewhere else in the RNA strand.
Pseudoknot-free secondary structures correspond to well nested structures.
They can be represented as a well-parenthesized expression where matching parentheses symbolize base pairs, with additional dots to symbolize unpaired nucleotides.

Given a sequence of nucleotides, the \textsc{RNA Folding} problem consists of finding the pseudoknot-free secondary structure with the minimum free-energy.
\textsc{RNA Folding} can be solved by a simple dynamic programming in time $O(n^3)$ where $n$ is the size of the sequence~\cite{Nussinov1980,Zuker1981}.
Since this result in the early 1980s, a lot of work has been devoted to propose new methods for secondary structure prediction.
Recently, the first truly subcubic algorithm for \textsc{RNA Folding} was proposed by Bringmann et al.~\cite{Bringmann16} and runs in deterministic time $O(n^{2.861})$ or randomized time $O(n^{2.825})$.
Some faster polynomial-time \emph{approximation} algorithms were later obtained~\cite{Saha17}.

When pseudoknots are allowed, the computational complexity of predicting how RNA folds gets a bit blurry.
The short answer would be to say that the folding prediction becomes NP-complete~\cite{LyngsoP00,Akutsu00,Ieong03}.
Observing that none of those three hardness constructions are ideal, the first one because the value of the free-energy is not fixed but specified as part of the input, and the other two, because they assume that only planar pseudoknots are legal, Lyngs\o{} gives two additional NP-hardness proofs \cite{Lyngso04}.
They work in seemingly very close models; both not too distant from Watson-Crick.
However, the NP-hardness in one model crucially needs that the alphabet size is unbounded (which is rather unsatisfactory), while the NP-hardness in the other model carries over to a binary alphabet.
When only restricted types of pseudoknots are allowed, dynamic programming still works and yields polynomial-time algorithms with worse running times than in the pseudoknot-free case \cite{Rivas99,LyngsoP00,Akutsu00,Chen09}.
Under the hypothesis that the pseudoknots may only form after the pseudoknot-free pairs, the $O(n^3)$-time complexity can be attained again~\cite{Jabbari08}.
In a simpler model, an approach based on maximum weighted matchings makes the folding prediction tractable for general pseudoknots~\cite{Tabaska98}.

\paragraph*{RNA design.}
In the inverse folding problem called \textsc{RNA Design}, one is given a secondary structure and has to find a sequence of bases which \emph{uniquely} folds into this structure; or report that such a sequence does not exist. 
The sequence must fold into this structure instead of any other structure.
In particular, in the Watson-Crick model, any other structure the sequence can fold into must have strictly fewer pairs.
If such a sequence exists, we call it a \emph{design} for the secondary structure.

This problem was introduced in the early 1990s in a paper which, to date, has over 2000 citations~\cite{Hofacker1994}.
The motivation to study this problem comes from the fact that the functions performed by particular RNA sequences are strongly influenced by the secondary structures these sequences fold into. Thus an important step towards designing RNA sequences that perform given functions is to be able to design sequences that fold into given secondary structures~\cite{ANDERSONLEE2016748}.
Surprisingly, the complexity of \textsc{RNA Design} is still unknown despite two decades of works and was explicitly stated as a major open problem~\cite{ANDRONESCU2004607,Condon03,Condon2006,hales:hal-01285499,DBLP:journals/corr/abs-1709-08088,Lyngso}. 
This is exceptional for a central problem in computational biology~\cite{hales:hal-01285499}.
Schnall-Levin et al. gave a NP-hardness proof for a more general problem~\cite{Schnall-LevinCB08}. 
However, to be applicable to \textsc{RNA Design}, the energy model would have to depend on the \textsc{3-Sat} instance in the reduction (hence, would be different for each instance) which is clearly not realistic (see the discussions in~\cite[Section 5]{hales:hal-01285499} or in~\cite{Lyngso}).

Solving \textsc{RNA Design} finds applications in multiple fields such as pharmaceutical research and biochemistry~\cite{hales:hal-01285499} as well as synthetic biology and RNA nanostructures \cite{Churkin2017}; the two latter areas aim at creating enhanced RNA with desirable properties.
It is also a major step towards functional RNA molecular design.
Therefore, there are many algorithms and software products\footnote{The following wikipedia page already references more than a dozen \url{https://en.wikipedia.org/wiki/List_of_RNA_structure_prediction_software\#Inverse_folding.2C_RNA_design}, last access: \today} solving the RNA inverse folding problem \cite{DBLP:journals/nar/Garcia-MartinCD13,Aguirre-Hernandez2007,ANDRONESCU2004607,Butterfoss2006ComputerbasedDO}.
Churkin et al. compare the main freewares solving \textsc{RNA Design} such as \texttt{RNAInverse}, \texttt{antaRNA}, and \texttt{RNAiFold} \cite{Churkin2017}.
All of them are either heuristics, or use meta-heuristics, or have an exponential running time.
Let us also mention the EteRNA project, an online game where (more than 100.000) players have to find a correct sequence given a structure~\cite{Lee14}.

Recently, Hale\v{s} et al. gave some sufficient conditions under which one can answer to the problem in polynomial time~\cite{hales:hal-01285499}.
Their main result is to show that if the structure is saturated, i.e., does not contain any unpaired letter, then a design --if it exists-- can be found in linear time by a greedy procedure.
In the case of saturated structures, the existence of a design is solely based on the maximum degree of a tree representing the structure. 
The authors also show that on smaller alphabets and general secondary structures, \textsc{RNA Design} is tractable.
This line of research was later continued by Jedwab et al.~\cite{DBLP:journals/corr/abs-1709-08088}. 
The authors presented an infinite family of designable structures containing unpaired letters. 
Again, the characterization of these structures is given in terms of trees. 

Following the name of the precoloring extension problem in graphs~\cite{BiroHT92}, let the \emph{extension} version (\designext) of the inverse folding problem be the same as \textsc{RNA Design} with the additional constraint that some indices of the RNA sequence should contain a specified base.
Lyngs\o{} observes that this assumption is biologically coherent: ``\emph{Most recent methods do allow for position specific constraints, where in addition to folding into the target structure the designed sequence is also required in certain positions to have a particular nucleotide}''~\cite{Lyngso}.
Indeed, in addition to the target structure, one has to force some bases at key positions to ensure that the RNA molecule possesses a given function.
Zhou et al. also propose a method to solve \textsc{RNA Design} where some positions within the sequence are constrained to certain bases~\cite{ZhouPVWZD13}.
Rodrigo et al. impose the presence of a certain sequence at a specific position in the structure~\cite{rodrigo2012}.
Borujeni et al. enforce the presence of a given subsequence (called the Shine-Dalgarno sequence) paired to a ``start codon'' to start the translation of RNA to proteins~\cite{borujeni2016}.
Furthermore, software libraries solving \textsc{RNA Design} allow those additional unary constraints and the instances of the EteRNA project contain immutable nucleotides.
Thus it appears that the design of RNA secondary structures is better captured by \designext than its restriction \textsc{RNA Design}.

\medskip

In this paper, 
we show that \designext is NP-hard in the simple Watson-Crick model of energy, suggesting that the same bound holds for more realistic energy models.

\paragraph*{Ideas of the reduction.}
The main reason the complexity of designing RNA secondary structures has been open for about twenty years is that it is difficult to create challenging structures for which the intended sequence will not fold into an undesired better structure; let alone to actually prove it.
It is considerably easier to exhibit an alternative better structure for a bad sequence.
Indeed, in the former case, one needs to argue over \emph{all} the structures compatible with the sequence, while in the latter, one just needs to find \emph{one} particular structure.
With that in mind, YES-instances of the starting NP-hard problem will be much more problematic to deal with than the NO-instances.

Our reduction is from \textsc{E3-SAT} (where all the clauses have exactly three literals).
Each clause gadget contains some unpaired bases and is surrounded by an \emph{arch} of nested parentheses.
The number of unpaired bases and the width of this arch are set so that if the three literals of the clause are unsatisfied (i.e. false), one can obtain a better structure by deleting the arch and matching the previously unpaired bases with other previously unpaired bases in the corresponding variable gadgets.
However, if only at most two literals of the clause are unsatisfied this rematching strategy ends up with a worst structure. 

To combat any improving rematching strategy for the sequences we want to interpret as satisfying assignments, we use arches of increasing widths to represent the variables.
This allows a simple counting argument to significantly prune the set of undesired rematchings.
Another key technical ingredient is to display the variable and clause gadgets interleaved in a carefully chosen order.
Interestingly, we also make use of the fact that saturated structures can be efficiently designed~\cite{hales:hal-01285499} in the correctness of our reduction.

\paragraph*{Robustness of the reduction.}
As established in the Watson-Crick energy model, our hardness result enjoys the following healthy properties.
We only need a 4-letter alphabet for the sequences, which naturally corresponds to the four nucleotides $A, U, C, G$.
This is optimal in light of the paper by Hale\v{s} et al. \cite{hales:hal-01285499} where the authors show the tractability of designing RNA secondary structures on an alphabet of size at most 3.
The way free-energy is computed is fixed (it can be thought as $-1$ for each base pair); hence, it is not part of the input and cannot be used to artificially encode a hard task.
We do not need pseudoknots --which make the folding prediction intractable-- to obtain the hardness. 
Watson-Crick being the simplest model, our result strongly suggests that RNA secondary structure design is hard in more authoritative models.

Let also note that the structures produced by our reduction are reasonably \emph{realistic}.  
They contain as building blocks nested parentheses surrounding some dots.
Interestingly, this structure is, as we mentioned, known as a stem-loop which is itself a building block of RNA structures.
Finally, we believe that the ideas developed in the reduction can be adapted to fit other energy models and will prove useful to show NP-hardness even when no element of the sequence is constrained to be a specified nucleotide.

\paragraph*{Organization.}
The rest of the paper is organized as follows.
In Section~\ref{sec:prelim}, we formally introduce all the required notions and define the problem \textsc{RNA Design (Extension)}.
In Section~\ref{sec:hardness}, we show our main contribution: even in the very simple Watson-Crick model, designing RNA secondary structures is NP-hard if the input structure comes with imposed bases at some specific positions.
In short, \designext is NP-hard.
In Section~\ref{sec:other}, we give simple algorithms with a complexity better than the brute-force for \textsc{RNA Design (Extension)}.
\ifconf Due to space constraints, some proofs (marked with a $\bigstar$) are in the appendix. \fi











\section{Preliminaries}\label{sec:prelim}
For a positive integer $n$, we denote by $[n]$ the set $\{1,2,\ldots,n\}$ of positive integers no greater than $n$. 
Given a word $w$ of length $n$ over an alphabet $\Sigma$, $w[i] \in \Sigma$ denotes the $i$-th letter of $w$ (for $i \in [n]$).

\paragraph*{Sequences and extensions.}
An RNA sequence is a word over the set of bases $\{A,C,G,U\}$.
A \emph{sequence} is a word over $\{1,2,3,4\}$, where $1$ represents $A$, $4$ represents $U$, $2$ represents $C$, and $3$ represents $G$.
This way, two letters can be \emph{paired} if they sum up to $5$.
We call \emph{base} an element of $\{1,2,3,4\}$.
A \emph{partial sequence} is a word over $\{1,2,3,4,?\}$.
An \emph{extension} of a partial sequence $w$ is a sequence $w'$ of the same length $n$ such that $\forall i \in [n]$, if $w[i] \neq \ ?$ then $w[i]=w'[i]$, and if $w[i] = \ ?$ then $w'[i] \in \{1,2,3,4\}$.

\paragraph*{Secondary structures.}
A \emph{pseudoknot-free secondary structure} (or \emph{structure} for short) is any word over the alphabet $\{(,),.\}$ such that if one removes all the $.$, the remaining word is a well parenthesized expression.
In what follows, we will always omit the adjective \emph{pseudoknot-free}. 
A well-parenthesized expression (or member of the Dyck language) is a word with the same number of $($ and $)$, and such that no prefix of the word have more $)$ than $($.
We call \emph{letter} an element of $\{(,),.\}$.
We refer to $.$ as an \emph{unpaired letter}, or an \emph{unmatched letter}, or simply a \emph{dot}, as opposed to $($ and $)$, which are \emph{paired}.
A structure is \emph{saturated} if it does \emph{not} contain any unpaired letter.

\paragraph*{Designs.}
In a well-parenthesized expression $E$, an opening parenthesis at index $i$ is said to be \emph{matched to} a closing parenthesis at index $j$ if $j$ is the smallest index to satisfy $j>i$ and that the multiset $\{E[i+1],E[i+2], \ldots, E[j-2],E[j-1]\}$ contains the same number of opening and closing parentheses.
We extend this definition to structures by ignoring the unpaired letters. 
A structure $S$ is \emph{compatible} with a sequence $w$, if they have the same length and for any indices $i<j \in [n]$ such that $S[i] = ($ is matched to $S[j]= \ )$, then $w[i]$ and $w[j]$ can be paired (i.e., $\{w[i],w[j]\} \in \{\{1,4\},\{2,3\}\}$).
A sequence $w$ is a \emph{design} for a structure $S$ if $S$ is compatible with $w$ and every other structure $S'$ compatible with $w$ has strictly more unpaired letters.
A partial sequence $w$ \emph{can be extended} to a design of $S$ if there is an extension $w'$ of $w$ which is a design for $S$.
We also say that a (partial) sequence $w$ \emph{labels} an index $i$ (or, by a slight abuse of language, \emph{a letter $l := S[i]$}) of a structure $S$ of the same size with (or \emph{by}) a base $b \in \{1,2,3,4,?\}$ if $w[i] = b$.

In \designext, one is given a structure $S$ and a partial sequence $w$ of the same length.
The goal is to decide if $w$ can be extended to a design for $S$.
The \textsc{RNA Design} problem can be seen as the special case when the partial sequence $w$ only contains $?$ symbols.
In words, no index of the structure $S$ is constrained to be labeled by a specific base of $\{1,2,3,4\}$.
In the introduction, we argued that \designext is perhaps more natural than its restriction \textsc{RNA Design}.

\begin{example}
$w=214??1?1343$ is a partial sequence.
$w'=21423121343$ is an extension of $w$.
$S=(()()(( . )))$ is a structure (since $(()()(()))$ is a well-parenthesized expression).
$S$ is compatible with $w'$.
However, $w'$ is \emph{not} a design for $S$ since $S'=(((())( . )))$ is also compatible with $w'$ and has the same number of unpaired letters (only one).
Actually, $w$ cannot be extended to a design of $S$ since, none of $\{21414121343,21441121343,21432121343\}$ is a design for $S$.
Observe that, in order to be compatible with $S$, the two first $?$ in $w$ have to get bases that can be paired while the third $?$ should be a $2$ to be paired with the following $3$.
The sequence $11423312424$ is a design for $S$.
\end{example}


\paragraph*{Tree representation of a structure.}
A structure $S$ can be seen as a rooted tree $T$, whose nodes are either pairs of matching parentheses (we call such nodes {\em paired}), or unmatched letters (we call them {\em unpaired}). The parent-child relation is defined by nestedness of parentheses/unmatched letters. Following Hale\v{s} {\em et al.} \cite{hales:hal-01285499}, for convenience we also add a special node which is a {\em virtual root} of $T$. 
Its role is to simplify working with structures which are not surrounded by parentheses.
Note that the children of each node are ordered and an unpaired node is always a leaf.

Every substructure of $S$, defined by a subtree of $T$ is itself called a \emph{subtree}.
Observe that \emph{not} every subword of $S$ is a subtree.
The {\em degree} of a node in $T$ is the number of its neighbors, excluding the unpaired ones (note that we count the parent of a node as a neighbor). 
Finally, by the \emph{degree} of a structure we mean the maximum degree of a node in its tree.

\section{Hardness of \designext}\label{sec:hardness}
The following lemma is intuitive and straightforward to prove.
We will use it repeatedly in order to prove our main theorem.
It says that a design induces designs in all the subtrees of the structure.

\begin{lemma}\ifconf[$\bigstar$]\fi\label{lem:substructure}
A design $w$ for a structure $S$ labels every \subtree of $S$ with a design.
\end{lemma}
\appendixproof{Lemma~\ref{lem:substructure}}
{
\begin{proof}
Assume that there is a \subtree $T$ of $S$ which is labeled by $w'$ and $w'$ is not a design.
Let $T'$ be a second structure compatible with $w'$, having at least as many paired letters as $T$.
The structure $S'$ obtained by replacing in $S$ the \subtree $T$ by $T'$ is another structure compatible with $w$ and having at least as many paired letters as $S$; a contradiction.
\end{proof}
}

We show the main result of the paper.
A first glimpse of the construction may consist of reading the dedicated paragraph in the introduction and going through Figure~\ref{fig:var-lit-gadget} to \ref{fig:no-variable-to-variable}.

\begin{theorem}\label{thm:main}
\designext is NP-complete. 
\end{theorem}

\begin{proof}
  \designext is in NP because the polynomial dynamic programming algorithm to solve \textsc{RNA Folding} in the papers \cite{Nussinov1980,Zuker1981} can be slightly adapted to test the uniqueness of the maximally matched structure.
  Therefore, one can guess (certificate of polynomial size) an extension of the partial sequence into a design (if such an extension exists) and check it with the tuned dynamic programming.
  
We reduce from the NP-hard problem \textsc{E3-Sat}, which is a variant of \textsc{3-Sat}, in which all clauses have exactly three distinct literals.
This problem remains NP-hard when each variable appears at most four times \cite{Tovey84}.
Let $\mathcal I = (X=\{x_1,~\cdots,x_n\}, \mathcal C=\{C_1,~\cdots,C_m\})$ be such an instance with $n$ variables and thus $m=\Theta(n)$ $3$-clauses.
We will build an equivalent instance $\mathcal J=(S,w)$ of \designext with a structure $S$ and a partial sequence $w$ both of length $N=\Theta(n^6)$.

Let $t := n^2$ and $y := (n+3m)t$.
In the structure $S$, for every variable and every literal we will introduce a gadget containing $t$ consecutive unpaired letters.
There will not be any other unpaired letter in $S$.
Hence $y = (n+3m)t$ represents the overall number of unpaired letters.
It might be useful to keep in mind that $n+m=\Theta(n) \ll t=\Theta(n^2) \ll y=\Theta(n^3)$.
For all the inequalities in the proof to hold, we assume that $n, m$ are greater that some large constant, say $1000$.
\paragraph*{Variable gadget.}
The gadget encoding a variable  $x_i$ is defined as follows:
$$\var{x_i} :=  \underbrace{((((((((((((((((}_{\text{length}~i(m+1)y} \underbrace{.........}_{\text{length}~t} \underbrace{))))))))))))))))}_{\text{length}~i(m+1)y} $$
where the opening parentheses are labeled by 1, the closing parentheses are labeled by 4 (see Figure~\ref{fig:var-gadget}).
We will refer to those parentheses as \emph{the arch of $x_i$}.
The next lemma indicates how the dots can be labeled in the variable gadgets.

A \emph{potential solution} is an extension of $w$ whose restriction to the subtree corresponding to one variable gadget is a design.
By Lemma~\ref{lem:substructure}, we know that a solution to the \designext instance $(S,w)$ has to be a potential solution.
\begin{lemma}\ifconf[$\bigstar$]\fi\label{lem:var-gadget}
In a potential solution, the dots in $\var{x_i}$ all receive label 2, or all receive label 3.
\end{lemma}
\appendixproof{Lemma~\ref{lem:var-gadget}}
{
\begin{proof}
Indeed, if one dot is labeled by 1 (resp. 4), then there would be a distinct structure, with the same number of pairs, that matches this dot to the first closing parenthesis (resp. to the last opening parenthesis).
Consider now the case where one dot is labeled by 2 and another dot is labeled by 3.
Then those two dots can be matched\footnote{By that slight abuse of language, we mean that the letters labeling those dots in one structure can be matched together to form a new structure (with more pairs). Let us also recall that we use the words \emph{matched} and \emph{paired} interchangeably.} together, yielding a structure with strictly more pairs.
\end{proof}
}

We interpret labeling all the dots of $\var{x_i}$ by $2$ to setting $x_i$ to true, and labeling all the dots by $3$ to setting $x_i$ to false. 
The dots in the variable gadgets will be the only letters of the structure $S$ which are not originally labeled by $w$.

\paragraph*{Clause gadget.}
In the clause gadgets, the structure is entirely labeled by $w$. 
Consider a 3-clause $C_j = \ell_a \lor \ell_b \lor \ell_c$ with $a < b < c$ and $\ell_i \in \{x_i,\neg x_i\}$ (for $i \in \{a,b,c\}$).
We define a \emph{literal gadget} for each literal of $C_j$.
For $\ell_a$ this gadget is denoted by $\lit{\ell_a}$ and is the same as $\var{x_a}$, where all the opening parentheses are labeled by 1, all the closing parentheses, by 4, and the dots are labeled by $2$ if the literal is positive and by $3$ if it is negative (see Figure~\ref{fig:lit-pos-gadget} and Figure \ref{fig:lit-neg-gadget}).
For $\ell_i \in \{\ell_b,\ell_c\}$, the literal gadget is denoted by $\litalt{j}{\ell_i}$ and is obtained from $\lit{\ell_i}$ by removal of $jy$ pairs of parentheses in the surrounding arch; so their number is only $(b(m+1)-j)y$ in $\litalt{j}{\ell_b}$ and $(c(m+1)-j)y$ in $\litalt{j}{\ell_c}$.

The whole clause $C_j$ is encoded by the clause gadget:
$$\cla{C_j} := \underbrace{\textbf{(\ldots(}}_{jy}\underbrace{\textbf{(\ldots(}}_{q} ((\litalt{j}{\ell_b})(~\underbrace{\textbf{(\ldots(}}_{jy}(\lit{\ell_a})\underbrace{\textbf{)\ldots)}}_{jy}))((\litalt{j}{\ell_c}))\underbrace{\textbf{)\ldots)}}_{q}\underbrace{\textbf{)\ldots)}}_{jy}$$
with $q := 3t-10(n+m)$. 
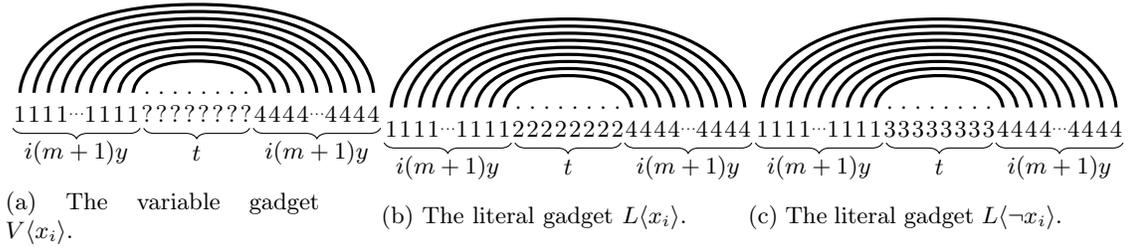
\begin{figure}
  \centering
  \begin{minipage}{0.28\textwidth}
    \centering
    \begin{tikzpicture}[scale=0.94]
      \def\r{0.2}
            \foreach \i in {0,1,...,3,5,6,...,8}{
      	\pgfmathsetmacro\j{5 - \i * \r}
      	\link{\i * \r}{0}{\j}{1.25 - 0.5 * \i * \r}
              \node at (\i * \r, -0.3) {$1$} ;
              \node at (\j, -0.3) {$4$} ;
            }
            \foreach \i in {4}{
      	\pgfmathsetmacro\j{5 - \i * \r}
      	\link{\i * \r}{0}{\j}{1.25 - 0.5 * \i * \r}
              \node at (\i * \r, -0.3) {\tiny{...}} ;
              \node at (\j, -0.3) {\tiny{...}} ;
            }
            \foreach \i in {9,...,16}{
            \node at (\i * \r, 0) {$.$} ;
            }
            \draw [decorate,decoration={brace,amplitude=4pt}] (3.25,-0.5) -- (1.75,-0.5) ;
            \node at (2.5,-0.85) {$t$} ;
           
            \foreach \i in {-0.2,3.2}{
              \begin{scope}[xshift=\i cm]
                \draw [decorate,decoration={brace,amplitude=4pt}] (1.9,-0.5) -- (0.1,-0.5) ;
                \node at (1,-0.85) {$i(m+1)y$} ;
              \end{scope}
            }
      
      \foreach \i in {9,...,16}{
        \node at (\i * \r, -0.3) {$?$} ;
      }
     \end{tikzpicture}
    \subcaption{The variable gadget $\var{x_i}$.}
    \label{fig:var-gadget}
  \end{minipage}
  \qquad
    \begin{minipage}{0.28\textwidth}
      \centering
      \begin{tikzpicture}[scale=0.94]
        \def\r{0.2}
              \foreach \i in {0,1,...,3,5,6,...,8}{
        	\pgfmathsetmacro\j{5 - \i * \r}
        	\link{\i * \r}{0}{\j}{1.25 - 0.5 * \i * \r}
                \node at (\i * \r, -0.3) {$1$} ;
                \node at (\j, -0.3) {$4$} ;
              }
              \foreach \i in {4}{
        	\pgfmathsetmacro\j{5 - \i * \r}
        	\link{\i * \r}{0}{\j}{1.25 - 0.5 * \i * \r}
                \node at (\i * \r, -0.3) {\tiny{...}} ;
                \node at (\j, -0.3) {\tiny{...}} ;
              }
              \foreach \i in {9,...,16}{
              \node at (\i * \r, 0) {$.$} ;
              }
              \draw [decorate,decoration={brace,amplitude=4pt}] (3.25,-0.5) -- (1.75,-0.5) ;
              \node at (2.5,-0.85) {$t$} ;
             
              \foreach \i in {-0.2,3.2}{
                \begin{scope}[xshift=\i cm]
                  \draw [decorate,decoration={brace,amplitude=4pt}] (1.9,-0.5) -- (0.1,-0.5) ;
                  \node at (1,-0.85) {$i(m+1)y$} ;
                \end{scope}
              }
        
        \foreach \i in {9,...,16}{
          \node at (\i * \r, -0.3) {$2$} ;
        }
      \end{tikzpicture}
    \subcaption{The literal gadget $\lit{x_i}$.}
    \label{fig:lit-pos-gadget}
    \end{minipage}
    \qquad
    \begin{minipage}{0.28\textwidth}
    \centering
    \begin{tikzpicture}[scale=0.94]
      \def\r{0.2}
            \foreach \i in {0,1,...,3,5,6,...,8}{
      	\pgfmathsetmacro\j{5 - \i * \r}
      	\link{\i * \r}{0}{\j}{1.25 - 0.5 * \i * \r}
              \node at (\i * \r, -0.3) {$1$} ;
              \node at (\j, -0.3) {$4$} ;
            }
            \foreach \i in {4}{
      	\pgfmathsetmacro\j{5 - \i * \r}
      	\link{\i * \r}{0}{\j}{1.25 - 0.5 * \i * \r}
              \node at (\i * \r, -0.3) {\tiny{...}} ;
              \node at (\j, -0.3) {\tiny{...}} ;
            }
            \foreach \i in {9,...,16}{
            \node at (\i * \r, 0) {$.$} ;
            }
            \draw [decorate,decoration={brace,amplitude=4pt}] (3.25,-0.5) -- (1.75,-0.5) ;
            \node at (2.5,-0.85) {$t$} ;
           
            \foreach \i in {-0.2,3.2}{
              \begin{scope}[xshift=\i cm]
                \draw [decorate,decoration={brace,amplitude=4pt}] (1.9,-0.5) -- (0.1,-0.5) ;
                \node at (1,-0.85) {$i(m+1)y$} ;
              \end{scope}
            }
      
      \foreach \i in {9,...,16}{
        \node at (\i * \r, -0.3) {$3$} ;
      }
    \end{tikzpicture}
    \subcaption{The literal gadget $\lit{\neg x_i}$.}
    \label{fig:lit-neg-gadget}
  \end{minipage}
    \caption{The variable and literal gadgets for $x_i$.
      If the ? in $\var{x_i}$ are labeled by 2 (resp.~3) --\emph{setting $x_i$ to true (resp.~false)}--, then $\var{x_i}$ can be entirely rematched to $\lit{\neg x_i}$ (resp.~$\lit{x_i}$).
      Note that $\lit{x_i}$ and $\lit{\neg x_i}$ do not depends of the truth assignment but only if they appear positively of negatively in a clause. }
\label{fig:var-lit-gadget}
\end{figure}
The outermost $jy$ opening parentheses of $\cla{C_j}$ are labeled by 1, forcing the corresponding $jy$ closing parentheses to be labeled by 4.
The next $q$ opening parentheses are labeled by 2, and their matching parentheses are labeled by 3.
The extra $jy$ opening parentheses surrounding $\lit{\ell_a}$ are labeled by \emph{4}, and their matching parentheses, by \emph{1}.
The label of the few remaining parentheses is specified in Figure~\ref{fig:clause-gadget}.

We will refer to the $jy+q$ outermost pairs of parentheses as \emph{the arch of $C_j$}.
We also call the first $jy$ pairs, the \emph{first layer} of the arch, and the next $q$ pairs, the \emph{second layer}.
Let $\mathcal A(q)_j$ denote the set of indices of the second layer in the clause gadget $\cla{C_j}$.
Let $\mathcal A(q)^2_j \subseteq \mathcal A(q)_j$ be the indices of the opening parentheses (labeled by $2$) and $\mathcal A(q)^3_j := \mathcal A(q)_j \setminus \mathcal A(q)^2_j$ be the indices of the closing parentheses (labeled by $3$).
Finally, let $\mathcal A(q) := \bigcup_{j \in [m]} \mathcal A(q)_j$.
\begin{figure*}
  \centering
\begin{tikzpicture}
\def\r{0.2}
\foreach \i in {1,...,6}{
	\pgfmathsetmacro\j{14 - \i * \r}
	\link{\i * \r}{0}{\j}{3 - 0.5 * \i * \r}
        \node at (\i * \r, -0.3) {$1$} ;
        \node at (\j, -0.3) {$4$} ;
}
\foreach \i in {7,...,10}{
	\pgfmathsetmacro\j{14 - \i * \r}
	\link{\i * \r}{0}{\j}{3 - 0.5 * \i * \r}
        \node at (\i * \r, -0.3) {$2$} ;
        \node at (\j, -0.3) {$3$} ;
}

\link{11 * \r}{0}{43 * \r}{1.5}
\link{12 * \r}{0}{26.5 * \r}{0.8}
\link{27.5 * \r}{0}{42 * \r}{0.8}

\link{44 * \r}{0}{59 * \r}{1}
\link{45 * \r}{0}{58 * \r}{0.8}

\node at (3.8,0.1) {$\litalt{j}{\ell_b}$} ;
\node at (6.95,0) {$\lit{\ell_a}$} ;
\node at (10.4,0.1) {$\litalt{j}{\ell_c}$} ;

\node at (11 *   \r, -0.3) {$1$} ;
\node at (12 *   \r, -0.3) {$2$} ;
\node at (26.5 * \r, -0.3) {$3$} ;
\node at (27.5 * \r, -0.3) {$3$} ;

\node at (28.5 * \r, -0.3) {$4$} ;
\node at (29.5 * \r, -0.3) {\tiny{...}} ;
\node at (30.5 * \r, -0.3) {$4$} ;
\node at (31.5 * \r, -0.3) {$2$} ;
\node at (38   * \r, -0.3) {$3$} ;
\node at (39   * \r, -0.3) {$1$} ;
\node at (40   * \r, -0.3) {\tiny{...}} ;
\node at (41   * \r, -0.3) {$1$} ;
\link{28.5 * \r}{0}{41 * \r}{.7}
\link{29.5 * \r}{0}{40 * \r}{.62}
\link{30.5 * \r}{0}{39 * \r}{.53}
\link{31.5 * \r}{0}{38 * \r}{.44}

\node at (42   * \r, -0.3) {$2$} ;
\node at (43   * \r, -0.3) {$4$} ;
\node at (44   * \r, -0.3) {$2$} ;
\node at (45   * \r, -0.3) {$2$} ;
\node at (58   * \r, -0.3) {$3$} ;
\node at (59   * \r, -0.3) {$3$} ;

\foreach \i in {0,12.63}{
  \begin{scope}[xshift=\i cm]
    \draw [decorate,decoration={brace,amplitude=4pt}] (1.25,-0.5) -- (0.1,-0.5) ;
    \node at (0.675,-0.85) {$jy$} ;
  \end{scope}
}

\foreach \i in {0,10.6}{
  \begin{scope}[xshift=\i cm]
    \draw [decorate,decoration={brace,amplitude=4pt}] (2.08,-0.5) -- (1.3,-0.5) ;
    \node at (1.7,-0.85) {$q$} ;
\end{scope}
}

\foreach \i in {0,2.085}{
  \begin{scope}[xshift=\i cm]
    \draw [decorate,decoration={brace,amplitude=4pt}] (6.17,-0.5) -- (5.65,-0.5) ;
    \node at (5.9,-0.85) {$jy$} ;
\end{scope}
}

\end{tikzpicture}
\caption{A $3$-clause gadget $\cla{C_j}$: $\ell_a \lor \ell_b \lor \ell_c$ with $a<b<c$.}
\label{fig:clause-gadget}
\end{figure*}
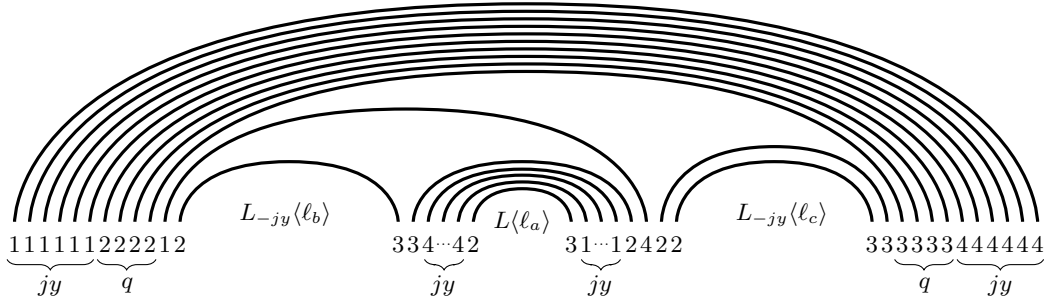

\paragraph*{Overall construction.}
We join all the gadgets in a binary tree of $\Theta(n+m)=\Theta(n)$ pairs of parentheses and height $\Theta(\log n)$ labeled such as illustrated in Figure~\ref{fig:overall-construction}.
The only requirement on this labeling is that there is no other way of fully matching the binary tree onto itself.
In other words, the labeling should be a design for the structure restricted to the parentheses of the binary tree.
We say that a pair of matched parentheses is \emph{labeled by $i$-$j$} (with $i+j=5$) if the opening parenthesis is labeled by $i$ (implying that the closing parenthesis has to be labeled by $j=5-i$). 
A possibility\footnote{Theorem 1 in \cite{HalesMPS15} shows that there are exponentially many possible labelings for the tree. For more details, see the proof of Lemma~\ref{lem:at-least-one-unpaired} in the present paper.} for labeling the binary tree is to use 1-4 for the outermost pair of parentheses and recursively use 2-3 and 3-2 for the two children of parentheses labeled by 1-4 or 4-1, and use 1-4 and 4-1 for the two children of parentheses labeled by 2-3 or 3-2.
We denote by $\mathcal T$ the set of indices of the letters in this binary tree.
At the ``leaves'' of the binary tree, we place the $n+m$ variable and clause gadgets.
The gadgets $\var{x_1}$, $\var{x_2}$, up to $\var{x_n}$ are placed from left to right. 
For $i \in [n-1]$, we reserve some room in between $\var{x_i}$ and $\var{x_{i+1}}$ for some clause gadgets, according to the following rule.
For each clause $C_j$ on variables $x_a$, $x_b$, $x_c$, with $a < b < c$, we insert the gadget $\cla{C_j}$ somewhere between $\var{x_b}$ and $\var{x_c}$; in other words, to the right of $\var{x_b}$ and to the left of $\var{x_c}$.
Obviously, such an ordering of the variable and clause gadgets can be found in polynomial time.  
The order of the clause gadgets that are between the same two consecutive variable gadgets $\var{x_i}$ and $\var{x_{i+1}}$ is not important and can be chosen arbitrarily.
As $n+m$ need not be a power of 2, there might be, as in Figure~\ref{fig:overall-construction}, some \emph{empty} ``leaves'' without a variable gadget nor a clause gadget.
We will show that the partial sequence $w$ can be extended into a design for the structure $S$ if and only if $\mathcal I$ is satisfiable.

The whole construction can be seen as simulating the following game, equivalent to \3sat, where your opponent has the more interesting role.
There are boxes with $3$ literals written on top of each box.
At the beginning of the game, an opponent, who does not want you to get rich, chooses a truth assignment of the variables appearing on the boxes.
Opening a box costs $2.99$\textcurrency\xspace of your favorite currency \textcurrency\xspace.	
The rules say that you can open at most one box.
Once you open a box, you find inside one object per literal.
You can turn this object into $1$\textcurrency\xspace if the literal is unsatisfied, and the object is worthless otherwise.
Knowing that, you will decide to open a box if and only if its three literals are unsatisfied (and win $3$\textcurrency\xspace$-$ $2.99$\textcurrency\xspace$=0.01$\textcurrency\xspace).
If you open a box with at least one satisfied literal, then you lose at least $2.99$\textcurrency\xspace$-$ $2$\textcurrency\xspace$=0.99$\textcurrency\xspace.
If the formula is satisfiable and your opponent is computationally almighty, he will choose a satisfying assignment.
And you will not win anything: you will decide not to open any box since it has a negative outcome.

In our case, \emph{opening a box by paying 2.99}\textcurrency\xspace corresponds to destroying, in a clause gadget, the $q$ pairs of innermost parentheses surrounding the three literal gadgets (together with some $\Theta(\log n)$ pairs of parentheses in $\mathcal T$ and an additional constant number within the clause gadget); and \emph{turning an object, found inside the box, associated to an unsatisfied literal $\ell_i$ into $1$\textcurrency\xspace} corresponds to fully pairing a variable gadget to a matching literal gadget.

\begin{figure*}
  \centering
\begin{tikzpicture}
\def\r{0.2}
\def\s{0.3}

\link{0}{0}{65 * \r}{10.5 * \s}

\link{1 * \r}{0}{32 * \r}{7 * \s}
\link{33 * \r}{0}{64 * \r}{7 * \s}
 
\link{2 * \r}{0}{16 * \r}{4.5 * \s}
\link{17 * \r}{0}{31 * \r}{4.5 * \s}
\link{34 * \r}{0}{48 * \r}{4.5 * \s}
\link{49 * \r}{0}{63 * \r}{4.5 * \s}

\link{3 * \r}{0}{8.5 * \r}{2.8 * \s}
\link{9.5 * \r}{0}{15 * \r}{2.8 * \s}
\link{18 * \r}{0}{23.5 * \r}{2.8 * \s}
\link{24.5 * \r}{0}{30 * \r}{2.8 * \s}
\link{35 * \r}{0}{40.5 * \r}{2.8 * \s}
\link{41.5 * \r}{0}{47 * \r}{2.8 * \s}
\link{50 * \r}{0}{55.5 * \r}{2.8 * \s}
\link{56.5 * \r}{0}{62 * \r}{2.8 * \s}

\node at (5.75 * \r,0.12) {$\var{x_1}$} ;
\node at (12.25 * \r,0.12) {$\var{x_2}$} ;
\node at (20.75 * \r,0.12) {$\cla{C_1}$} ;
\node at (27.25 * \r,0.12) {$\cla{C_2}$} ;

\begin{scope}[xshift=32 * \r cm]
\node at (5.75 * \r,0.15) {$\var{x_3}$} ;
\node at (12.25 * \r,0.15) {$\cla{C_3}$} ;
\node at (20.75 * \r,0.15) {$\var{x_4}$} ;
\end{scope}

\node at (0 *    \r, -0.3) {$1$} ;
\node at (1 *    \r, -0.3) {$2$} ;
\node at (2 *    \r, -0.3) {$1$} ;
\node at (3 *    \r, -0.3) {$2$} ;
\node at (8.5 *  \r, -0.3) {$3$} ;
\node at (9.5 *  \r, -0.3) {$3$} ;
\node at (15 *   \r, -0.3) {$2$} ;
\node at (16 *   \r, -0.3) {$4$} ;
\node at (17 *   \r, -0.3) {$4$} ;
\node at (18 *   \r, -0.3) {$2$} ;
\node at (23.5 * \r, -0.3) {$3$} ;
\node at (24.5 * \r, -0.3) {$3$} ;
\node at (30   * \r, -0.3) {$2$} ;
\node at (31   * \r, -0.3) {$1$} ;
\node at (32   * \r, -0.3) {$3$} ;
\node at (33   * \r, -0.3) {$3$} ;
\node at (34 *   \r, -0.3) {$1$} ;
\node at (35 *   \r, -0.3) {$2$} ;
\node at (40.5 * \r, -0.3) {$3$} ;
\node at (41.5 * \r, -0.3) {$3$} ;
\node at (47 *   \r, -0.3) {$2$} ;
\node at (48 *   \r, -0.3) {$4$} ;
\node at (49 *   \r, -0.3) {$4$} ;
\node at (50 *   \r, -0.3) {$2$} ;
\node at (55.5 * \r, -0.3) {$3$} ;
\node at (56.5 * \r, -0.3) {$3$} ;
\node at (62   * \r, -0.3) {$2$} ;
\node at (63   * \r, -0.3) {$1$} ;
\node at (64   * \r, -0.3) {$2$} ;
\node at (65   * \r, -0.3) {$4$} ;

\end{tikzpicture}
\caption{The overall picture with 4 variables and 3 clauses. $C_1$ and $C_2$ are on variables $x_1,x_2,x_3$ while $C_3$ is on variables $x_1$ or $x_2$ and variables $x_3,x_4$. }
\label{fig:overall-construction}
\end{figure*}
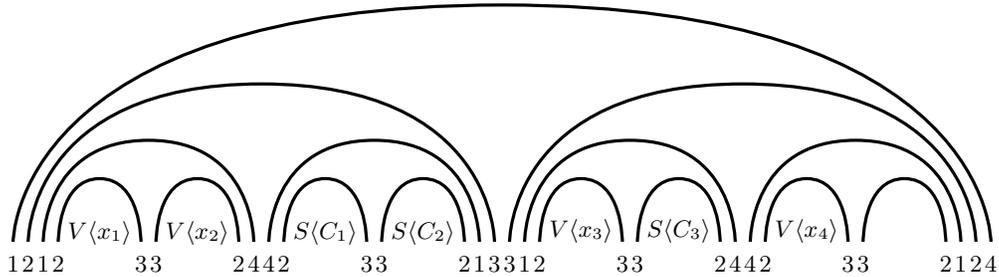

\paragraph*{$\mathcal I$ unsatisfiable implies $\mathcal J$ has no design extension.} 
Assume that instance $\mathcal I$ is not satisfiable. 
By Lemma~\ref{lem:var-gadget}, we already observed that every potential solution corresponds to a truth assignment (via the interpretation: dots to 2 $\equiv$ true, dots to 3 $\equiv$ false).
For any potential solution $w'$, let $A$ be the corresponding variable assignment.  
By assumption, there is a clause which is not satisfied by $A$; suppose it is the clause $C_j = \ell_a \lor \ell_b \lor \ell_c$.

For two structures $S_1,S_2$  compatible with the same sequence $w'$, we say that a parenthesis or a dot at index $u$ in $S_1$ is \emph{rematched in $S_2$} to a parenthesis or a dot  at index $v$ in $S_1$ if $u$ and $v$ are the indices of matching parentheses in $S_2$ (but not in $S_1$).
Similarly, we say that the letter at index $u$ (resp.~the index $u$ itself) is \emph{rematched} to the letter at index $v$ (resp.~the index $v$ itself).
If the nature of the structures $S_1$ and $S_2$ is obvious from the context, we will not precise them. 

We exhibit a structure $S'$ compatible with $w'$ and with more paired letters than $S$ (see Figure~\ref{fig:rematch}).
In this paragraph and the next one, the role of $S_1$ is played by $S$ and the role of $S_2$, by $S'$.
The $jy$ opening parentheses of the first layer of the arch of $C_j$ are rematched to the last $jy$ closing parentheses of the arch of $x_b$, while the $jy$ closing parentheses of the first layer of the arch of $C_j$ are rematched to the first $jy$ opening parentheses of the arch of $x_c$.
The letters whose indices are in $\mathcal A(q)_j$ (second layer) become unpaired.
We fully rematch $\var{x_b}$, $\var{x_a}$, and $\var{x_c}$, to $\litalt{j}{\ell_b}$, $\lit{\ell_a}$, and $\litalt{j}{\ell_c}$, respectively.
It is only possible since all those three literals are unsatisfied by $A$; which means that, for any $i \in \{a,b,c\}$, if the dots in $\var{x_i}$ are labeled by 2 (resp. 3), then the dots in $\lit{\ell_i}$ or $\litalt{j}{\ell_i}$ are labeled by 3 (resp. 2).
So, those dots can be matched with each other.
Observe that the extra arch above $\lit{\ell_a}$ absorbs the first $jy$ opening parentheses of the arch of $x_b$ and the last $jy$ closing parentheses of the arch of $x_c$.
Whereas the first layer of the arch of $\cla{C_j}$ absorbs the last $jy$ closing parentheses of the arch of $x_b$ and the first $jy$ opening parentheses of the arch of $x_c$. 

Rematching those six sets of $t$ consecutive dots incurs a win of $3t$ pairs.
Let us now count the number of pairs in $S$ that we lose.
We have to break at most $6 \lceil \log(n+m) \rceil$ pairs in $\mathcal T$ (indices in the binary tree), so that the six gadgets $\var{x_a}$, $\var{x_b}$, $\var{x_c}$, and $\lit{\ell_a}$, $\litalt{j}{\ell_b}$, $\litalt{j}{\ell_c}$ in $\cla{C_j}$ can be rematched with each other.
Those pairs that we break are all the parentheses in $\mathcal T$ in the paths going from those six gadgets to the root of the binary tree.
Actually removing \emph{all} the parentheses of $\mathcal T$ would still work. 
We also broke the $q=3t-10(n+m)$ pairs of indices in $\mathcal A(q)$, plus the $6$ pairs of parentheses in $\cla{C_j}$ which are not part of an arch.
The rest of $S'$ is matched as in $S$. 
This new structure has at least $3t-(3t-10(n+m)+6) - 6 \lceil \log(n+m) \rceil = 10(n+m) - 6(\lceil \log(n+m) \rceil +1) > 0$ more pairs.
Hence, $S$ partially labeled by $w$ cannot be extended into a design.
\begin{figure*}
  \centering
\begin{tikzpicture}[scale=0.6,yscale=0.9]
  \def\r{0.2}
\foreach \i in {-27,...,-24}{
	\pgfmathsetmacro\j{- 11.5 - \i * \r}
	\link{\i * \r}{0}{\j}{6 + \i * \r - 0.2}
        \node at (\i * \r, -0.3) {\tiny{$4$}} ;
        \node at (\j, -0.3) {\tiny{$1$}} ;
}
\foreach \i in {-20,...,-12}{
	\pgfmathsetmacro\j{- 4 - \i * \r}
	\link{\i * \r}{0}{\j}{- 0.8 * \i * \r - 1.4}
        \node at (\i * \r, -0.3) {\tiny{$1$}} ;
        \node at (\j, -0.3) {\tiny{$4$}} ;
}
\foreach \i in {70,...,82}{
	\pgfmathsetmacro\j{34 - \i * \r}
	\link{\i * \r}{0}{\j}{12 - 0.7 * \i * \r}
        \node at (\i * \r, -0.3) {\tiny{$1$}} ;
        \node at (\j, -0.3) {\tiny{$4$}} ;
}
 
\foreach \i in {2,...,4}{
	\pgfmathsetmacro\j{14 - \i * \r}
	\link{\i * \r}{0}{\j}{4 - 0.7 * \i * \r}
        \node at (\i * \r, -0.3) {\tiny{$1$}} ;
        \node at (\j, -0.3) {\tiny{$4$}} ;
}
\foreach \i in {5,...,10}{
	\pgfmathsetmacro\j{14 - \i * \r}
	\link{\i * \r}{0}{\j}{4 - 0.7 * \i * \r}
        \node at (\i * \r, -0.3) {\tiny{$2$}} ;
        \node at (\j, -0.3) {\tiny{$3$}} ;
}

\link{11 * \r}{0}{43 * \r}{1.7}
\link{12 * \r}{0}{26.5 * \r}{0.8}
\link{27.5 * \r}{0}{42 * \r}{1.15}

\link{44 * \r}{0}{59 * \r}{1}
\link{45 * \r}{0}{58 * \r}{0.8}

\node at (3.8,0.1) {\tiny{$\litalt{j}{\ell_b}$}} ;
\node at (6.95,-0.1) {\tiny{$\lit{\ell_a}$}} ;
\node at (10.4,0.1) {\tiny{$\litalt{j}{\ell_c}$}} ;

\node at (11 *   \r, -0.3) {\tiny{$1$}} ;
\node at (12 *   \r, -0.3) {\tiny{$2$}} ;
\node at (26.5 * \r, -0.3) {\tiny{$3$}} ;
\node at (27.5 * \r, -0.3) {\tiny{$3$}} ;

\node at (28.5 * \r, -0.3) {\tiny{$4$}} ;
\node at (29.5 * \r, -0.3) {\tiny{$4$}} ;
\node at (30.5 * \r, -0.3) {\tiny{$4$}} ;
\node at (31.5 * \r, -0.3) {\tiny{$2$}} ;
\node at (38   * \r, -0.3) {\tiny{$3$}} ;
\node at (39   * \r, -0.3) {\tiny{$1$}} ;
\node at (40   * \r, -0.3) {\tiny{$1$}} ;
\node at (41   * \r, -0.3) {\tiny{$1$}} ;
\link{28.5 * \r}{0}{41 * \r}{1}
\link{29.5 * \r}{0}{40 * \r}{.82}
\link{30.5 * \r}{0}{39 * \r}{.63}
\link{31.5 * \r}{0}{38 * \r}{.44}

\node at (42   * \r, -0.3) {\tiny{$2$}} ;
\node at (43   * \r, -0.3) {\tiny{$4$}} ;
\node at (44   * \r, -0.3) {\tiny{$2$}} ;
\node at (45   * \r, -0.3) {\tiny{$2$}} ;
\node at (58   * \r, -0.3) {\tiny{$3$}} ;
\node at (59   * \r, -0.3) {\tiny{$3$}} ;

\begin{scope}[red,yshift=-0.5cm,yscale=0.8]
  \link{0 * \r }{0}{2 * \r }{-0.1}
  \link{-1 * \r }{0}{3 * \r }{-0.3}
  \link{-2 * \r }{0}{4 * \r }{-0.5}
  \link{-3 * \r }{0}{13 * \r }{-0.8}
  \link{-4 * \r }{0}{13.7 * \r }{-1}
  \link{-5 * \r }{0}{14.4 * \r }{-1.2}
  \link{-6 * \r }{0}{15.1 * \r }{-1.4}
  \link{-7 * \r }{0}{15.8 * \r }{-1.6}
  \link{-8 * \r }{0}{16.5 * \r }{-1.8}

  \link{-10 * \r }{0}{19 * \r }{-2.1}

  \link{-12 * \r }{0}{21.5 * \r }{-2.4}
  \link{-13 * \r }{0}{22.2 * \r }{-2.6}
  \link{-14 * \r }{0}{22.9 * \r }{-2.8}
  \link{-15 * \r }{0}{23.6 * \r }{-3}
  \link{-16 * \r }{0}{24.3 * \r }{-3.2}
  \link{-17 * \r }{0}{25 * \r }{-3.4}
  
    \link{-18 * \r }{0}{28.5 * \r }{-3.7}
    \link{-19 * \r }{0}{29.5 * \r }{-3.9}
    \link{-20 * \r }{0}{30.5 * \r }{-4.1}

    \link{-24 * \r }{0}{32.5 * \r }{-4.5}
    \link{-25 * \r }{0}{33 * \r }{-4.7}
    \link{-26 * \r }{0}{33.5 * \r }{-4.9}
    \link{-27 * \r }{0}{34 * \r }{-5.1}

    \link{-29 * \r }{0}{35 * \r }{-5.4}

    \link{-30.5 * \r }{0}{36 * \r }{-5.7}
    \link{-31.5 * \r }{0}{36.5 * \r }{-5.9}
    \link{-32.5 * \r }{0}{37 * \r }{-6.1}
    \link{-33.5 * \r }{0}{37.5 * \r }{-6.3}

    \link{70 * \r }{0}{68 * \r }{-0.1}
    \link{71 * \r }{0}{67 * \r }{-0.3}
    \link{72 * \r }{0}{66 * \r }{-0.5}

    \link{73 * \r }{0}{57 * \r }{-1}
    \link{74 * \r }{0}{56.5 * \r }{-1.2}
    \link{75 * \r }{0}{56 * \r }{-1.4}
    \link{76 * \r }{0}{55.5 * \r }{-1.6}
    \link{77 * \r }{0}{55 * \r }{-1.8}
    \link{78 * \r }{0}{54.5 * \r }{-2}
    \link{79 * \r }{0}{54 * \r }{-2.2}
    \link{80 * \r }{0}{53.5 * \r }{-2.4}
    \link{81 * \r }{0}{53 * \r }{-2.6}
    \link{82 * \r }{0}{52.5 * \r }{-2.8}

    \link{85 * \r }{0}{51.5 * \r }{-3.1}

    \link{88 * \r }{0}{50.5 * \r }{-3.4}
    \link{89 * \r }{0}{50 * \r }{-3.6}
    \link{90 * \r }{0}{49.5 * \r }{-3.8}
    \link{91 * \r }{0}{49 * \r }{-4}
    \link{92 * \r }{0}{48.5 * \r }{-4.2}
    \link{93 * \r }{0}{48 * \r }{-4.4}
    \link{94 * \r }{0}{47.5 * \r }{-4.6}
    \link{95 * \r }{0}{47 * \r }{-4.8}
    \link{96 * \r }{0}{46.5 * \r }{-5}
    \link{97 * \r }{0}{46 * \r }{-5.2}

    \link{98 * \r }{0}{41 * \r }{-5.5}
    \link{99 * \r }{0}{40 * \r }{-5.7}
    \link{100 * \r }{0}{39 * \r }{-5.9}
\end{scope}

\node at (-29   * \r, 1.5) {\tiny{$\var{x_a}$}} ;
\node at (-29   * \r+0.05, 0) {\tiny{...}} ;

\node at (-9.5   * \r, 2.3) {\tiny{$\var{x_b}$}} ;
\node at (-10   * \r, 0) {\tiny{...}} ;

\node at (85  * \r, 2.7) {\tiny{$\var{x_c}$}} ;
\node at (85   * \r, 0) {\tiny{...}} ;

\node at (36  * \r, 4.2) {\tiny{$\cla{C_j}$}} ;

\end{tikzpicture}
\caption{Suppose a clause $C_j$ on variables $x_a$, $x_b$, and $x_c$, with $a < b < c$, \emph{is not satisfied by the extension $w'$}. In red (light gray) is how we build a structure $S'$ with more pairs than $S$ (in black).}
\label{fig:rematch}
\end{figure*}
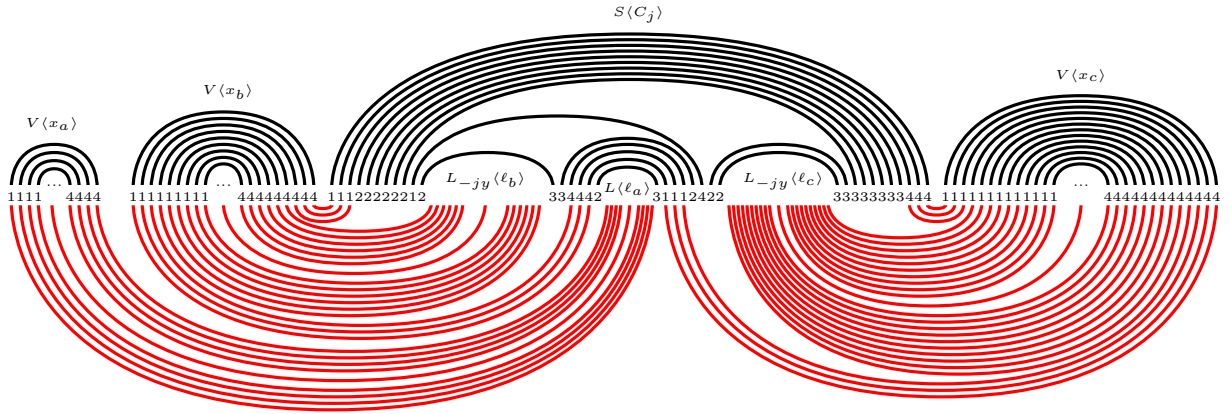
\paragraph*{$\mathcal I$ satisfiable implies $\mathcal J$ has a design extension.} 
Let $A$ be a satisfiable assignment and let $w'$ be the extension of $w$ corresponding to $A$.
We show that $w'$ is a design for $S$.
For the sake of contradiction, assume that there is a structure $S' \neq S$ compatible with $w'$, having at least as many pairs as $S$.
Let us take $S'$ maximally matched, i.e, such that there is no structure $S''$ compatible with $w'$ with strictly more pairs than $S'$.

\begin{lemma}\ifconf[$\bigstar$]\fi\label{lem:at-least-one-unpaired}
$S'$ has to match at least one letter which is unpaired in $S$.
\end{lemma}
\appendixproof{Lemma~\ref{lem:at-least-one-unpaired}}
{
\begin{proof}
  Assume that $S'$ does not match any unpaired letters in $S$.
  As $S'$ has at least the same number of paired letters as $S$, it implies that $S'$ matches \emph{exactly} the same letters (meaning the same indices) as $S$.
  Let $R'$ (resp.~$R$) be the structure obtained by restricting $S'$ (resp.~$S$) to the paired letters of $S$.
  Let $\hat w$ be $w'$ restricted to those paired letters.
  By construction, $R'$ and $R$ are two distinct saturated structures both compatible with $\hat w$.
  In the proof of Theorem 1 in \cite{HalesMPS15}, the authors show that every saturated structure with degree at most $4$ has a design (in fact, many designs); and that this design can be found (in linear time) by a \emph{greedy} labeling.
  The greedy labeling is \emph{any} labeling which does not assign labels $i$-$j$ to a child of paired parentheses labeled by $j$-$i$ and avoids labeling two siblings with the same (oriented) pair $i$-$j$.    
  Observe that $R$ has degree at most $4$ and that $\hat w$, which we fully specified in the above construction, respects those two rules.
  Thus, $\hat w$ is a design for $R$; a contradiction to the existence of $R'$. 
\end{proof}
}

The following lemma is straightforward and proves useful to argue about the quality of a structure reachable from a partially built structure.
It is based on a simple counting argument.
For any $i \in [4]$, we denote by $\#(i,w)$ the number of occurrences of $i$ in the word $w$. 

\begin{lemma}\ifconf[$\bigstar$]\fi\label{lem:imbalance}
If a structure contains $(R)$ as a subtree (that is, the opening and closing parentheses around $R$ match) and the labeling $\hat w$ of $R$ is complete, then, for any $i \in [4]$ and integer $x$, $|\#(i,\hat w) - \#(5-i,\hat w)| > x$ implies that more than $x$ letters will remain unpaired.
\end{lemma}
\appendixproof{Lemma~\ref{lem:imbalance}}
{
\begin{proof}
Let $\hat I$ be the set of (consecutive) indices of letters labeled by $\hat w$. 
Because of the surrounding parentheses, a base with index in $\hat I$ has to be matched with a base with index also in $\hat I$.
If, in $\hat w$, the number of $i$ exceeds the number of $5-i$ by more than $x$, then more than $x$ bases $i$ will not find a pairing $5-i$.
\end{proof}
}

\begin{figure*}
\centering
\begin{tikzpicture}[yscale=1]
\def\r{0.4}
\def\s{0.4}
\def\sp{0.3}
\def\t{3.95}	
\def\n{4}
\foreach \i in {1,...,\n}{
  \pgfmathsetmacro\h{(1 + \i * \i - \i + 3 * \i) * \r}
  \pgfmathsetmacro\k{(\i * \i + \i + 3 * \i + 4) * \r}
  \pgfmathsetmacro\m{\h / 2 + \k / 2}
  \node at (\m,0.12) {$\var{x_\i}$} ;
  \foreach \j in {1,...,\i}{
    \pgfmathsetmacro\d{(1.3 +  0.8 * \i) * (1 - 0.17 * \j)}
    \begin{scope}[color=blue]
   \linkb{\h + \j * \r}{\m}{\k - \j * \r}{\d * \s}   
    \end{scope}
    \node at (\h + \j * \r,-0.2) {$1$} ;
    \node at (\k - \j * \r,-0.2) {$4$} ;
  }
}
\node[fill,circle,red,inner sep=-0.05cm] (a) at (20.5 * \r,0.5) {} ;
\node (ta) at (21 * \r,0.5) {$2$} ;
\draw[red,thick,dashed] (a) to [bend left=20] ++(1.5,0.5) node[right] {\textcolor{black}{$3?$}} ;
\draw[red,thick,dashed] (a) to [bend right=20] ++(-1.5,0.5) node[left] {\textcolor{black}{$3?$}} ;

\begin{scope}[yshift=-4cm]
\coordinate (d1) at (4 * \r,3) {} ;
\coordinate (d2) at (4.5 * \r,3) {} ;
\coordinate (d3) at (5.5 * \r,2) {} ;
\coordinate (d4) at (7.5 * \r,2) {} ;
\coordinate (d5) at (8.5 * \r,3) {} ;
\coordinate (d6) at (9.5 * \r,3) {} ;
\coordinate (d7) at (11.5 * \r,1) {} ;
\coordinate (d8) at (13.5 * \r,1) {} ;
\coordinate (d9) at (15.5 * \r,3) {} ;
\coordinate (d10) at (16.5 * \r,3) {} ;
\coordinate (d11) at (19.5 * \r,0) {} ;
\coordinate (d12) at (21.5 * \r,0) {} ;
\coordinate (d13) at (24.5 * \r,3) {} ;
\coordinate (d14) at (25.5 * \r,3) {} ;
\coordinate (d15) at (28.5 * \r,0) {} ;
\coordinate (d16) at (29.5 * \r,1) {} ;
\coordinate (d17) at (31.5 * \r,1) {} ;
\coordinate (d18) at (32.5 * \r,0) {} ;
\coordinate (d19) at (35.5 * \r,3) {} ;
\coordinate (d20) at (36 * \r,3) {} ;

\foreach \i [count = \j from 2] in {1,...,19}{
\draw[very thick] (d\i) -- (d\j) ;
}
\foreach \i in {0,...,3}{
\node (e\i) at (1,\i) {$\i u$} ;
}
\foreach \i in {1,...,3}{
\draw[very thin,dotted] (1.3,\i) -- (36 * \r,\i) ;
}
\draw[->,thick] (1.3,-0.1) -- (1.3,3.2) ;
\draw[->,thick] (1.2,0) -- (18.5 * \r,0) node[below] {position of the $3$ paired to the $2$ in $\var{x_3}$} -- (36.5 * \r,0) ;

\fill[opacity=0.15] (4 * \r,0) -- (4.5 * \r,0) -- (4.5 * \r,\t) -- (4 * \r,\t) -- cycle ;
\fill[opacity=0.15] (5.5 * \r,0) -- (7.5 * \r,0) -- (7.5 * \r,\t) -- (5.5 * \r,\t) -- cycle ;
\fill[opacity=0.15] (8.5 * \r,0) -- (9.5 * \r,0) -- (9.5 * \r,\t) -- (8.5 * \r,\t) -- cycle ;
\fill[opacity=0.15] (11.5 * \r,0) -- (13.5 * \r,0) -- (13.5 * \r,\t) -- (11.5 * \r,\t) -- cycle ;
\fill[opacity=0.15] (15.5 * \r,0) -- (16.5 * \r,0) -- (16.5 * \r,\t) -- (15.5 * \r,\t) -- cycle ;
\fill[opacity=0.15] (24.5 * \r,0) -- (25.5 * \r,0) -- (25.5 * \r,\t) -- (24.5 * \r,\t) -- cycle ;
\fill[opacity=0.15] (29.5 * \r,0) -- (31.5 * \r,0) -- (31.5 * \r,\t) -- (29.5 * \r,\t) -- cycle ;
\fill[opacity=0.15] (35.5 * \r,0) -- (36 * \r,0) -- (36 * \r,\t) -- (35.5 * \r,\t) -- cycle ;
\end{scope}
\end{tikzpicture}
\caption{Why pairing a letter in $\var{x_i}$, not paired in $S$, to a letter in $\var{x_{i'}}$ with $i \neq i'$ cannot give a sufficiently paired structure. 
A single blue edge represents an arch of thickness $u:=(m+1)y$.
The $y$-axis corresponds to a lower bound of the imbalance $|\#(1,\hat w)-\#(4,\hat w)|$ where $\hat w$ is the subsequence surrounded by the new pair depicted by the red dashed edge.
The gray areas mark positions where a $3$ can actually be present.
In those regions, $|\#(1,\hat w)-\#(4,\hat w)|$ is greater than $u$, so this pairing necessarily yields a worse structure than $S$.}
\label{fig:no-variable-to-variable}
\end{figure*}
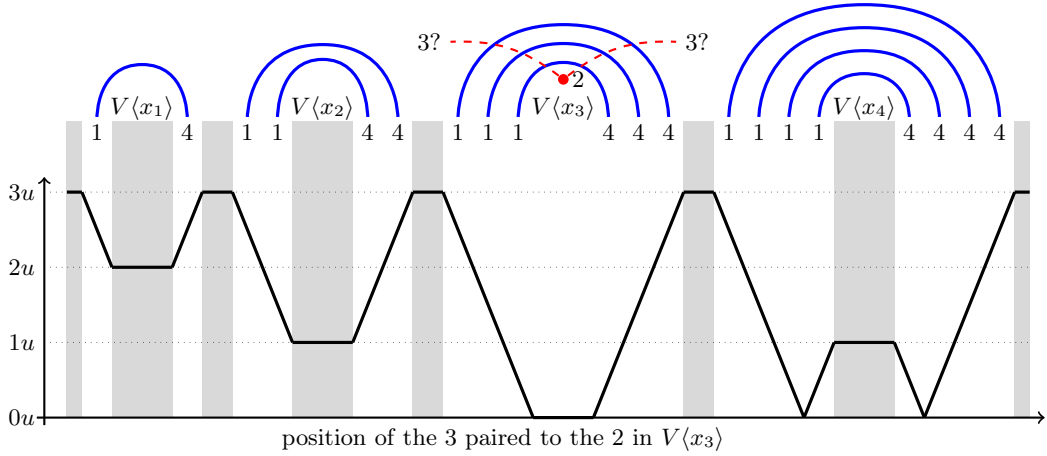
Let $D_i$ be the set of the dots contained in $\var{x_i}$ and in all the occurrences of $\lit{\ell_i}$ and $\litalt{j}{\ell_i}$ (with $\ell_i \in \{x_i,\neg x_i\}$ and some $j \in [m]$).

\begin{lemma}\ifconf[$\bigstar$]\fi\label{lem:xi_to_xi}
In $S'$, a base labeling a dot of $D_i$ can only be matched to a base labeling a dot of~$D_i$.
\end{lemma}
\appendixproof{Lemma~\ref{lem:xi_to_xi}}
{
\begin{proof}
  What is illustrated in Figure~\ref{fig:no-variable-to-variable} is in fact more general, and extends from the variable gadgets to the variable \emph{plus} the literal gadgets.
  Suppose a base labeling a dot of $D_i$ is matched to a base labeling a dot of $D_{i'}$ with $i \neq i'$.
  Then, such a pair would surround a subsequence $\hat w$ in $w'$ with $|\#(1,\hat w)-\#(4,\hat w)| > |i-i'|(m+1)y > (m+1)y = u > y$, provided the two bases do not appear within the same clause gadget.
  Remember that, in the gadget for the clause $C_j$, the extreme literals (corresponding to the second and third literals) are deprived of $jy$ matching parentheses in their arch, while the middle literal (corresponding to the first literal) has an additional $jy$ pairs of parentheses in its arch labeled 4-1 instead of 1-4.
  This explains why $jy$ does not appear in the upper bound of the imbalance.  
  The inequality $|\#(1,\hat w)-\#(4,\hat w)| > y$ holds for the three subcases: variable-variable, variable-clause, and clause-clause (where X-Y says that the first base appears in an X gadget and the second base appears in a distinct Y gadget), and no matter which literal of the clause (first, second, or third) contains the base.

  Now, if a base labeling a dot of $D_i$ in a literal gadget is matched to a base labeling a dot of $D_{i'}$ in the \emph{same} clause gadget $\cla{C_j}$ (but a different literal gadget), then such a pair would surround a subsequence $\hat w$ in $w'$ with $|\#(1,\hat w)-\#(4,\hat w)| > (|i-i'|(m+1)-j)y > |m+1-j|y > y$.

  Finally, if a base labeling a dot of $D_i$ is matched to a base which is \emph{not} labeling a dot of a $D_{i'}$, then the imbalance $|\#(1,\hat w)-\#(4,\hat w)|$, in the word $\hat w$ surrounded by this pair, is larger, for some $j \in [m]$, than $(i(m+1)-j)y > (m+1-m)y > y$.
  
  Recall that $S$ has only $y$ unpaired letters.
  By Lemma~\ref{lem:imbalance}, in all those cases, $S'$ would have at least $y+1$ unpaired letters; a contradiction.
\end{proof}
}

We will apply Lemma~\ref{lem:imbalance} again to argue that the second layer of the clause arches cannot be significantly rematched.
Let $\mathcal T'$ be $\mathcal T$ augmented with the constant number per clause gadget of indices corresponding to parentheses not in any arch.

\begin{lemma}\ifconf[$\bigstar$]\fi\label{lem:arch-A(q)}
In $S'$, a letter with index in $\mathcal A(q)^2_j$ (resp. $\mathcal A(q)^3_j$) can only be matched to a letter of $\mathcal T' \cup \mathcal A(q)^3_j$ (resp. $\mathcal T' \cup \mathcal A(q)^2_j$).  
\end{lemma}
\appendixproof{Lemma~\ref{lem:arch-A(q)}}
{
\begin{proof}
  First, in $S'$, $\mathcal A(q)^2_j$ cannot be rematched to $\mathcal A(q)^3_{j'}$ with $j \neq j'$.
  Such a match would indeed surround a subsequence $\hat w$ in $w'$ with $|\#(1,\hat w)-\#(4,\hat w)| > |j-j'|y$.
  By Lemma~\ref{lem:imbalance} that would imply that $S'$ has strictly fewer pairs than $S$.
  Second, matching a letter with index in $\mathcal A(q)^2_j$ to a 3 outside of $\mathcal T' \cup \bigcup_{j'} \mathcal A(q)^3_{j'}$ would mean to match it to a 3 labeling a dot in $S$.
  For this case, we can conclude similarly to the proof of Lemma~\ref{lem:xi_to_xi}. 
\end{proof}
}

Let $i \in [n]$ be such that a dot of $D_i$ labeled by 2 is matched in $S'$ to a dot of $D_i$ labeled by 3.
By Lemma~\ref{lem:at-least-one-unpaired} and Lemma~\ref{lem:xi_to_xi}, this index exists.
Since $D_i$ contains several literal gadgets but only one variable gadget $\var{x_i}$, at least one endpoint of this pair is in a clause gadget.
Let $j \in [m]$ be the index of this clause.
None of the pairs of parentheses of $\mathcal A(q)_j$ can be present in $S'$; otherwise the matching would cross. 
As $|\mathcal T'|=\Theta(n)$ and $q=\Theta(t)=\Theta(n^2)$, Lemma~\ref{lem:arch-A(q)} implies that most of those $q$ pairs in $S$ are unpaired in $S'$; only a negligible $O(n)$ of the corresponding letters could be rematched in $\mathcal T'$.   

Of the three literals of $C_j$, at most two are not satisfied by $A$.
Let $k \in [n]$ be the index of a satisfied literal in $C_j$.
The number of pairs of parentheses in $S$ destroyed in $S'$ is at least $q-O(n)=3t-O(n)$.
At best, $2t$ ($t$ per unsatified literal) new pairs are formed in $S'$ by linking literal gadgets in $\cla{C_j}$ to the corresponding variable gadgets.
This still incurs a deficit of $t-O(n)$ pairs.
The dots in the literal gadget of $x_k$ in $\cla{C_j}$ have to be rematched, since otherwise $S'$ is not maximal: reversing locally $S'$ to $S$ would provide a structure with strictly more pairs and would not create a crossing.
In other words, the structure $S''$ obtained from $S'$ by replacing the parentheses with at least one endpoint in $\cla{C_j}$ by the parentheses of $S$ in the gadget of $C_j$ and the two variable gadgets corresponding to the unsatisfied literals would have $t-O(n)>0$ more pairs than $S'$.

By Lemma~\ref{lem:xi_to_xi}, the dots in the literal gadget of $x_k$ in $\cla{C_j}$ can only be rematched to another clause gadget $\cla{C_{j'}}$ containing the opposite literal of $x_k$.
By the same argument as for $\cla{C_j}$, this rematching costs at least $3t-O(n)$ parentheses in $\mathcal A(q)_{j'}$ (so $6t-O(n)$ in total).
And only $5t$ new pairs can be obtained; this is the case if the four literals in the clauses $C_j$ and $C_{j'}$ which are not on the variable $x_k$ are unsatisfied and rematched to the corresponding variable gadgets\footnote{Observe that this case is not even possible, since otherwise the clause $C_{j'}$ would not be satisfied by $A$, so the actual deficit of this rematching strategy is even $2t-O(n)$. Although, a deficit of $t-O(n)$ was good enough.}.
The deficit of this rematching is $6t-5t-O(n)=t-O(n)>0$ (since we assumed that $n,m$ are large enough).
Thus reversing $S'$ to $S$ locally (in the two clause gadgets of $C_j$ and $C_{j'}$ and the five variable gadgets) would provide a structure with more pairs than $S'$, contradicting its maximality. 
\end{proof}



\section{Algorithmic results}\label{sec:other}
In this section we show that the trivial $O^*(4^n)$-time algorithm for the \textsc{RNA Design} problem can be significantly improved by analyzing the tree representation of the input sequence. 

Consider a structure $S$ and its tree representation $T$. Let us define two families of subtrees that can be found in $T$ (we follow the notation used by Hale\v{s} et al. \cite{hales:hal-01285499}). 
By $m_5$ we denote a node of degree more than 4. By $m_{3\circ}$ we denote a node with at least one unpaired child, and degree greater than 2. We will use the following result by Hale\v{s} et al. \cite{hales:hal-01285499} (note that both $m_5$ and $m_{3 \circ}$ do not denote a specific subtree, but rather infinite families of subtrees).

\begin{theorem}[Hale\v{s} et al. \cite{hales:hal-01285499}] \label{forbidden}
If $S$ is designable, then it contains neither $m_5$ nor $m_{3\circ}$.
\end{theorem}


\begin{theorem}\ifconf[$\bigstar$]\fi\label{thm:algo}
\designext can be solved in time:
\begin{compactenum}[(i) ]
\item $\sqrt{3}^n \cdot n^{O(1)}$, where $n$ is the length of the input structure,
\item $2^s \cdot n^{O(1)}$, where $s$ is the number of unlabeled elements in the input structure,
 \end{compactenum}
 using polynomial space.
\end{theorem}
\appendixproof{Theorem~\ref{thm:algo}}
{
\begin{proof}
Let $S$ be the input structure of length $n$ and with $s$ unlabeled elements. Let $T$ be the tree representation of $S$ and let  $r$ be the virtual root of $T$. 
By \autoref{forbidden} we know that $r$ has at most 4 children, each being either a matching pair of parentheses, or an unpaired letter. In the first step, we branch into all possible labelings of the unlabeled children of $r$.
If there are no more unlabeled nodes, we check in polynomial time if the obtained labeling is a design of $S$, and if it extends the predefined partial labeling.

Consider a non-leaf node $v$ of $T$, which is labeled, but all its children are unlabeled. Note that since $v$ is not a leaf, it corresponds to a pair of matching parentheses.

We now want to branch into all possible labelings of children of $v$. However, in some cases we can prune the search tree, if we know that the current partial solution is not extendable to a design.
By \autoref{forbidden}, we know that there are only very few possibilities of how the children of $v$ look like. Either all of them are paired and then $v$ has at most 3 children (otherwise we obtain $m_5$), or $v$ has some unpaired children and at most one paired child (otherwise we obtain $m_{3\circ}$).
Without loss of generality assume that $v$ is labeled with 1-4 (all other cases are symmetric).

\begin{figure}[h!]
\begin{center}
\begin{tikzpicture}
\def\r{0.15}

\link{0 * \r}{0}{10 * \r}{1}
\link{1 * \r}{0}{3 * \r}{0.6}
\link{4 * \r}{0}{6 * \r}{0.6}
\link{7 * \r}{0}{9 * \r}{0.6}

\altlink{0 * \r}{0}{10 * \r}{0.8}
\altlink{1 * \r}{0}{3 * \r}{0.5}
\altlink{4 * \r}{0}{9 * \r}{0.5}
\altlink{6 * \r}{0}{7 * \r}{0.3}

\node at (0 *   \r, -0.3) {$1$} ;
\node at (10 *   \r, -0.3) {$4$} ;
\node at (1 * \r, -0.3) {$1$} ;
\node at (3 * \r, -0.3) {$4$} ;
\node at (4 * \r, -0.3) {$2$} ;
\node at (6 * \r, -0.3) {$3$} ;
\node at (7 * \r, -0.3) {$2$} ;
\node at (9 * \r, -0.3) {$3$} ;
\end{tikzpicture}
\begin{tikzpicture}
\def\r{0.15}

\link{0 * \r}{0}{10 * \r}{1}
\link{1 * \r}{0}{3 * \r}{0.6}
\link{4 * \r}{0}{6 * \r}{0.6}
\link{7 * \r}{0}{9 * \r}{0.6}

\altlink{0 * \r}{0}{1 * \r}{0.3}
\altlink{3 * \r}{0}{10 * \r}{0.8}
\altlink{4 * \r}{0}{6 * \r}{0.6}
\altlink{7 * \r}{0}{9 * \r}{0.6}

\node at (0 *   \r, -0.3) {$1$} ;
\node at (10 *   \r, -0.3) {$4$} ;
\node at (4 * \r, -0.3) {$2$} ;
\node at (6 * \r, -0.3) {$3$} ;
\node at (1 * \r, -0.3) {$4$} ;
\node at (3 * \r, -0.3) {$1$} ;
\node at (7 * \r, -0.3) {$3$} ;
\node at (9 * \r, -0.3) {$2$} ;
\end{tikzpicture}
\begin{tikzpicture}
\def\r{0.15}

\link{0 * \r}{0}{20 * \r}{1}
\link{9 * \r}{0}{11 * \r}{0.6}
\altlink{0 * \r}{0}{9 * \r}{0.6}
\altlink{11 * \r}{0}{20 * \r}{0.6}
\node at (2 * \r, -0) {$\circ$} ;
\node at (8 * \r, -0) {$\circ$} ;
\node at (5 * \r, -0) {$\cdots$} ;
\node at (12 * \r, -0) {$\circ$} ;
\node at (18 * \r, -0) {$\circ$} ;
\node at (15 * \r, -0) {$\cdots$} ;

\node at (0 * \r, -0.3) {$1$} ;
\node at (20 * \r, -0.3) {$4$} ;
\node at (9 * \r, -0.3) {$4$} ;
\node at (11 * \r, -0.3) {$1$} ;
\end{tikzpicture}
\begin{tikzpicture}
\def\r{0.15}

\link{0 * \r}{0}{10 * \r}{1}
\altlink{5 * \r}{0}{10 * \r}{0.6}

\node at (2 * \r, -0) {$\circ$} ;
\node at (3.5 * \r, -0) {...} ;
\node at (5 * \r, -0) {$\circ$} ;
\node at (6.5 * \r, -0) {...} ;
\node at (8 * \r, -0) {$\circ$} ;

\node at (0 *   \r, -0.3) {$1$} ;
\node at (10 *   \r, -0.3) {$4$} ;
\node at (5 * \r, -0.3) {$1$} ;
\end{tikzpicture}
\begin{tikzpicture}
\def\r{0.15}

\link{0 * \r}{0}{10 * \r}{1}
\altlink{0 * \r}{0}{10 * \r}{0.6}
\altlink{4 * \r}{0}{6 * \r}{0.4}
\node at (2 * \r, -0) {$\circ$} ;
\node at (4 * \r, -0) {$\circ$} ;
\node at (6 * \r, -0) {$\circ$} ;
\node at (8 * \r, -0) {$\circ$} ;

\node at (0 *   \r, -0.3) {$1$} ;
\node at (10 *   \r, -0.3) {$4$} ;
\node at (4 * \r, -0.3) {$2$} ;
\node at (6 * \r, -0.3) {$3$} ;
\end{tikzpicture}
\end{center}

\caption{Some labelings which cannot be extended to a design, because they can be folded into some other structure with at least as many paired letters as in $S$ (the first two pictures correspond to Case I, the third one to Case II, and the last two to Case III).}
\label{fig:wrong-cases}
\end{figure}
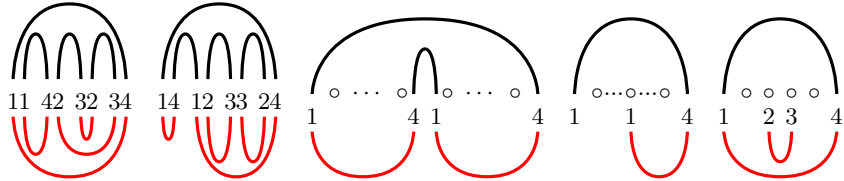


\noindent \textbf{Case I.} First, consider the case that $v$ has $d \leq 3$ children and all of them are paired. It is easy to verify that the only possible labelings of the children of $v$ are: 1-4, 2-3, 3-2 (in any ordering, each of them may appear only once, see \autoref{fig:wrong-cases} for some examples).
This gives the recursion \[F(n) \leq \frac{3!}{(3-d)!} F(n-2d),\]
where $F(n)$ denotes the complexity of the discussed algorithm for the input structure of length $n$.
The worst-case is achieved for $d=1$ and has complexity $F(n) = O^*(\sqrt{3}^n)=O(1.7321^n)$.

\noindent \textbf{Case II.} 
Now consider the case that $v$ has one paired child and $d \geq 1$ unpaired ones. We observe that the paired child of $v$ can be labeled with 1-4 only, while the unpaired children can get either 2 or 3, but all of them must receive the same label (again, see \autoref{fig:wrong-cases}).
This gives the recursion \[F(n) \leq 2 F(n-2-d).\]
The worst case is achieved for $d = 1$ and has complexity $F(n) = O^*(\sqrt[3]{2}^n)=O(1.2600^n)$.

\noindent \textbf{Case III.} 
Finally, consider the case that $v$ has $d \geq 1$ unpaired children and no paired ones. Let us call such a node {\em bad}. We observe that all children of $v$ must receive the same label, either 2 or 3 (again, see \autoref{fig:wrong-cases}), thus the recursion for this case is \[F(n) \leq 2F(n-d),\]
which gives the complexity bound $F(n) = O^*(2^n)$ (achieved for $d=1$).

However, we can show that this case cannot happen too often. We say that a node of $T$, which has at least two paired children, is {\em good}. Let $T'$ be a tree constructed from $T$ by removing all unpaired nodes, and contracting all induced paths into single edges (thus we remove nodes of degree 2).
Moreover, if the virtual root $r$ of $T$ has degree 1, we remove it and assume that $T'$ is rooted at the only child of $r$.
It is easy to observe that $T'$ has the following properties:
\begin{enumerate}[(a)]
\item every node of $T'$ is also a node of $T$,
\item the root of $T'$ has at most 4 children,
\item every inner node of $T'$ has 2 or 3 children,
\item every good vertex of $T$ is an inner node of $T'$,
\item every bad vertex of $T$ is a leaf of $T'$.
\end{enumerate}
Let $z$ be the number of leaves in $T'$, clearly the number of bad vertices in $T$ is at most $z$.
Since every inner node of $T'$ has at most 3 children, we observe that the number of inner nodes in $T'$ is at least $(z-4)/2$, so this gives us a lower bound for the number of good nodes in $T$.
Thus, for every two bad nodes (up to a constant number exceptional ones), there exists a good node (which is not shared with any other pair of bad nodes). More formally speaking, we can partition the set of all but a constant number of bad nodes into a family $A$ of two-elements sets,  and define an injective mapping from $A$ to the set of good nodes of $T$.
So, if we consider labeling children of two bad nodes (with $d_1$ and $d_2$ unpaired children, respectively) and $d \geq 2$ children of a good node at the same time, we obtain the recursion
\[
F(n) \leq 2^2 \cdot 6 \cdot F(n-d_1-d_2-2d).
\] 
The worst case-complexity for this case is achieved for $d_1=d_2=1$ and $d=2$, which gives us $F(n) \leqslant 24F(n-1-1-4)=24F(n-6)$, so $F(n) = O^*(\sqrt[6]{24}^n) = O(1.6984^n)$.
Since there is only a constant number of bad nodes which are not paired with good nodes, the blow-up in complexity is also a constant and can be ignored in $O(\cdot)$-notation.

The correctness of the described procedure is clear and follows from the fact that we only discard labeling which cannot appear in any design.
The running time of the procedure is determined by the complexity of the worst-case branching, which appears in Case I for $d=1$, and thus the running time can be bounded by $F(n) = O^*(\sqrt{3}^n)=O(1.7321^n)$.

Note that the above recursive procedure is completely oblivious to the initial partial labeling of $S$. The only place where we make use of it is the final checking. In our second approach we will only construct partial labelings which extend the pre-labeling.

First, observe that if $S$ has a matching pair and one of its elements is already labeled, the label of the other element is also uniquely determined. Thus we can assume that each node of $T$ is either labeled (i.e., if it is a paired node, then both parentheses are already labeled), or unlabeled.

The cases we consider are the same as in the first algorithm, but now the size of the problem is $s$, the number of unlabeled elements in the input structure. 

\noindent \textbf{Case I.} Let $d$ be the number of paired children of $v$, $p$ of which are unlabeled. We have $1 \leq p \leq d \leq 3$. Considering all cases, we observe that the worst case is achieved for $d=p=1$. It is described by the recursion 
\[
F'(s) \leq 3F'(s-2),
\]
and its complexity is $F'(s) = \sqrt{3}^s \cdot n^{O(1)},$
where $F'(s)$ is the complexity of the discussed algorithm for an input structure with $s$ unlabeled elements.

\noindent \textbf{Case II.} Recall that the labeling of the paired child of $v$ must be the same as the labeling of $v$. Thus, without loss of generality, we can assume that the paired child of $v$ is labeled. Also, if at least one of unpaired children is labeled, we can extend the labeling to all other unlabeled children. Suppose that $v$ has $d \geq 1$ unpaired children We obtain the recursion
\[
F'(s) \leq 2F'(s-d).
\]
Its worst-case is achieved for $d=1$, and the complexity is $F'(s) = 2^s \cdot n^{O(1)}$.

\noindent \textbf{Case III.} This case is analogous to the previous one -- all $d$ unpaired children of $v$ get the same label.
The recursion for this case is  
\[
F'(s) \leq 2F'(s-d),
\]
with the worst case achieved for $d=1$ and the complexity bound $F'(s) = 2^s \cdot n^{O(1)}$.

The correctness is straightforward and the complexity bound for the whole procedure is $F'(s) = 2^s \cdot n^{O(1)}$.
\end{proof}
}
Finally, let us remark that Hale\v{s} et al. \cite{hales:hal-01285499} give a complete characterization of saturated structures (i.e., ones without unpaired elements), which have a design. This characterization implies a polynomial-time algorithm for the \textsc{RNA design} problem on such structures. Using bottom-up dynamic programming on the tree representation on the input structure, we can adapt this procedure to the more general \designext problem.

\begin{observation}
\designext is tractable on saturated structures. \qed
\end{observation}

\ifconf
\else
\section*{Acknowledgments}
Thanks to Valia Mitsou and Yann Ponty for valuable discussions and comments.
\fi



\bibliographystyle{abbrv}
\bibliography{main}

\ifconf
\appendix
\newpage\section{Appendix}
\appendixProofText
\fi

\end{document}